\documentclass[11pt]{amsart}  
\usepackage{times}
\usepackage{hyperref}
\usepackage{a4wide}

\title{Bounds for binary codes relative to pseudo-distances of $k$ points}

\author{Christine Bachoc, Gilles Z\'emor}

\date{\today}
\newtheorem{defi}{Definition}[section]
\newtheorem{definition}[defi]{Definition}
\newtheorem{proposition}[defi]{Proposition}
\newtheorem{theorem}[defi]{Theorem}
\newtheorem{remark}[defi]{Remark}

\newtheorem{lemma}[defi]{Lemma}

\newcommand{\F}{{\mathbb{F}}} 

\newcommand{\R}{{\mathbb{R}}}

\newcommand{\W}{{\overline{\mathcal W}}}

\newcommand{\rr}{\overline{r}}

\newcommand{\Aut}{\operatorname{Aut}}

\newcommand{\card}{\operatorname{card}}

\newcommand{\ux}{\underline{x}}

\newcommand{\daff}{d^{\operatorname{aff}}}
\newcommand{\dkaff}{d_{k-1}^{\operatorname{aff}}}

\begin{document}

\begin{abstract}
We apply Schrijver's semidefinite programming method to obtain
improved upper bounds on generalized distances and list decoding radii
of binary codes.
\end{abstract}

\maketitle

\section{Introduction}

Let $H_n:=\F_2^n$ denote the binary Hamming space, endowed with the
Hamming distance. One of the longstanding problems of coding theory
is to find estimates for the maximum cardinality $A(n,d)$ of a
code $C\subset H_n$  with the constraint that the Hamming
distance of any pair of distinct elements of $C$ is at least equal to
$d$.
The best known upper bound for $A(n,d)$ is obtained with the
so-called linear programming method, due to Philippe Delsarte,
and is the optimal value of a linear program (LP for short) 
(\cite{Del1}, \cite[Chapter 9]{CS}). 
Because linear programs come with efficient algorithms, this method 
yields good numerical bounds for given parameters
$(n,d)$. Moreover, close to optimal explicit  feasible solutions have
been found from which upper bounds in the form of explicit functions
of $n$ and $d$ have been derived \cite{L}, as well as an upper bound in the
asymptotic range \cite{MRRW}. After these significant achievements, the
subject fell into a  period of about twenty years during which nothing
really new was discovered, until A. Schrijver in \cite{S} obtained 
improved upper bounds on $A(n,d)$ for some small values of the 
parameters $(n,d)$, using semidefinite programming.
Although these improvements are numerically not all
that impressive, the method behind them introduces genuinely new ideas. 
In order to explain them, it is good to go back to Delsarte's method. 
Let us recall that the variables of the Delsarte linear
program represent the distribution of the Hamming distance in the
constrained code. More precisely, let
\begin{equation*}
x_i:=\frac{1}{\card(C)}\card\{(x,y)\in C^2 : d(x,y)=i\}.
\end{equation*}
Then the main idea is to observe that these variables satisfy certain
linear inequalities, the non trivial ones being related to the
Krawtchouck polynomials $K^n_k(x)$, namely, for $0\leq k\leq n$,
\begin{equation*}
\sum_{i=0}^nK_k^n(i) x_i \geq 0.
\end{equation*}

Schrijver's new idea \cite{S} is to exploit constraints on 
{\em triples} of points
$(x,y,z)\in C^3$ rather than deal only with pairs. It turns out that
the natural constraints are semidefinite positive (SDP) instead of
linear. The variables of the program are 
\begin{equation*}
x_{a,b,c}:=\frac{1}{\card(C)} \card\{(x,y,z)\in C^3:
d(y,z)=a, d(x,z)=b, d(x,y)=c\}
\end{equation*}
and the  SDP constraints take the form
\begin{equation*}
\sum_{a,b,c} x_{a,b,c}S(a,b,c)\succeq 0,
\end{equation*}
where $\succeq 0$ stands for ``is positive semidefinite'',
for some symmetric matrices $S(a,b,c)$. These SDP constraints are
closely related to the action of the symmetric
group $S_n$ on the functional space $\R^{H_n}$; more precisely, each
$S_n$-irreducible module occuring in $\R^{H_n}$ gives rise to an SDP
inequality with  matrices of size equal to its multiplicity. It should be noted
that the full group of automorphisms
$\Aut(H_n)$ acts multiplicity free on the same space $\R^{H_n}$, and
that it is the true reason why in the case of the Delsarte method, the
constraints are linear. 

The aim of the present paper is to show  that the Schrijver method can be
used not only to strengthen the LP bounds, but also to give
bounds for other problems, to which the LP method does not apply. 
Indeed, in recent years several generalizations of the Hamming
distance, in the form of functions (we will call them
pseudo-distances) 
of $k\geq 3$ elements of $H_n$ have attracted attention. We consider here three
such functions $f(x_1,\dots,x_k)$, namely the generalized Hamming
distance $d(x_1,\dots,x_k)$, the radial distance $r(x_1,\dots,x_k)$ and the
average radial distance $\rr(x_1,\dots,x_k)$. They share the crucial property
of being invariant by the action of the automorphism group $\Aut(H_n)$ of
the Hamming space. 

The generalized Hamming {\em weights} of linear codes were introduced by 
Ozarow and Wyner \cite{OW} in view of cryptographic
applications related to the so-called wire-tap channel. 
The concept was later made popular for its own sake by Wei \cite{W}.
The notion was extended to the non linear setting in~\cite{CLZ} in
order to derive bounds on generalized weights. The generalized
Hamming {\em distance} $d(x_1,\dots,x_k)$ of $k$ points is the number of
coordinates where the $k$ points are not all equal. Thus $d(x_1,x_2)$
is the classical Hamming distance. In \cite{CLZ}, the authors derive
bounds for generalized distances, focusing on asymptotics,
which are analogs of the classical Hamming, Plotkin and
Elias-Bassalygo bounds. In the case of linear codes the best known 
asymptotic upper bounds were obtained in~\cite{ABL}.

The radial distance and the average radial distance are related to 
the notion of list decoding. 
The {\em radial distance} or {\em radius} of $k$ elements is the
smallest radius of a Hamming ball that contains the $k$ points. If a code $C$
has the property that the radius of any $k$-tuple of pairwise distinct
points is at least equal to some value $r$, then any ball of the
Hamming space of radius $r-1$ intersects $C$ in at most $k-1$
points. Thus a decoding procedure that outputs every codeword
at distance at most $r-1$ of any given received vector yields a {\em list}
of codewords of cardinality at most $k-1$. 
The search for large codes with given minimum $k$-radius is
also studied in the litterature as the quest for dense {\em multiple
  packings}: indeed, a code of minimal $k$-radius $r$
provides a packing of balls (centered at the codewords, of
radius $r-1$) such that any element of $H_n$ belongs to at most $k-1$ balls.
These notions have a long standing history, going back to
problems in Euclidean geometry and to early coding theory. 
They came back into the limelight some ten years ago when Sudan 
discovered his now famous
algorithm for list decoding of Reed-Solomon codes \cite{Su}.
Blinovskii~\cite{Bli1} establishes asymptotic bounds on the maximal number of
elements of a code with given minimal radius: in the process he
defines an auxiliary quantity, the {\em average radius} of $k$
elements that we will  also investigate.

In general, we are given a function $f$ from $H_n^k$ into the set of
non-negative integers, and we denote by
$A_{k-1}(n,f,m)$ the maximal number of
elements that a binary code $C$ can have under the constraint that
$f(x_1,\dots,x_k)\geq m$ for every $k$-tuple of pairwise distinct
codewords. Our goal is to show that the SDP method gives good upper
bounds for $A_2(n,f,m)$ for modest values of $n$, when compared with the
classical bounds. Our results provide strong motivation for the
development of the SDP method, which is far from being at the same
stage of achievement as the LP method. 
%In particular it is not yet clear how to deal with $k$-tuples in general. 

The paper is organized as follows: Section 2 provides a description
of the orbits of $\Aut(H_n)$ acting on $H_n^k$. This preliminary
task is essential since the pseudo-distances we are dealing with only
depend on these orbits. Section 3 recalls the definitions and basic
properties of the three particular functions we consider. Section 4 defines the
code invariants associated to these functions and recalls their
significance for applications. Section 5 settles the ``classical''
bounds. These bounds already appear in the litterature (\cite{CLZ},
\cite{Bli1}, \cite{Bli2}) but not in the
precise form needed here: either they are settled only for linear
codes, or the concern is in the asymptotic setting
and they are not as tight as they can be for small parameters.
Section 6 recalls the SDP method of \cite{S} using the language of
group representation, i.e. following \cite{Ba}, \cite{V1}, \cite{V2}. 
Section 7 provides some numerical results.

\section{The orbits of $Aut(H_n)$ acting on $H_n^{k}$}\label{orbits}

The automorphism
group of the binary Hamming space $H_n:=\F_2^n$, denoted by $\Aut(H_n)$,  is the  semi-direct
product of the group of translations by elements of $H_n$ with the
group of permutations on the $n$ coordinates.
The group $\Aut(H_n)$ acts two-point homogeneously on $H_n$, which 
means that the orbits of $\Aut(H_n)$ acting on $H_n^2$ are
characterized by the Hamming distance. In other words 
\begin{equation*}
(x,y)\sim_{\Aut(H_n)} (x',y') \Leftrightarrow d(x,y)=d(x',y').
\end{equation*}
Here $(x,y)\sim_{\Aut(H_n)} (x',y')$ stands for: there exists $g\in
\Aut(H_n)$
such that $g(x)=x'$ and $g(y)=y'$. We want to study the action of
$\Aut(H_n)$ on $k$-tuples $(x_1,\dots,x_k)\in H_n^{k}$. 
We introduce:

\begin{definition} For $\underline{x}=(x_1,\dots,x_{k})\in H_n^{k}$, and for $u\in
  \F_2^{k}$, let 
\begin{equation*}
n_u(\underline{x}):=\card\{\ j,\  1\leq j\leq n : ((x_1)_j,\dots,(x_k)_j)=u\}
\end{equation*}
and let the ``weight distribution'' of $\underline{x}$ be defined by
\begin{equation*}
\mathcal{W}(\underline{x}):=(n_u(\underline{x}))_{u\in \F_2^{k}}.
\end{equation*}
For $u\in \F_2^{k}$, the word obtained from $u$ by flipping zeros and
ones, will be denoted by $\overline{u}$. In other words
$\overline{u}=u+1^k$. One of $\{u,\overline{u}\}$ has the form $0w$
with $w\in \F_2^{k-1}$. Let 
\begin{equation*}
n_w(\underline{x}):=n_{0w}(\underline{x})+n_{1\overline{w}}(\underline{x}).
\end{equation*}
The ``symmetrized weight distribution'' of $\underline{x}$ is defined by:
\begin{equation*}
\overline{\mathcal{W}}(\underline{x}):=(n_w(\underline{x}))_{w\in \F_2^{k-1}}
\end{equation*}

\end{definition}

\noindent{\bf Remarks:}
\begin{enumerate}
\item It is nice to identify $\underline{x}$ with the $(k,n)$
  matrix  $M(\ux)$ whose $i$-th line equals $x_i$. Then $n_u(\underline{x})$ is the number of columns of $\underline{x}$ which are
equal to $u$:
\begin{equation*}
M(\ux)=
\begin{array}{lll}
x_1 &= 000\dots 0& \dots\dots\\
x_2 &= 111\dots 1& \dots\dots\\
\ \vdots &=\quad\vdots\qquad\vdots&\\
x_k &= \underbrace{111\dots 1}_{n_u(\underline{x})}& \dots\dots
\end{array}
\end{equation*}
\item We have $\sum_{u\in \F_2^{k}} n_u(\ux)=\sum_{w\in \F_2^{k-1}}
  n_w(\underline{x})=n$. 
\end{enumerate}

\begin{proposition}
\begin{equation*}
\underline{x}\sim_{\Aut(H_n)} \underline{y} \Leftrightarrow \W(\underline{x})=\W(\underline{y}).
\end{equation*}
\end{proposition}

\begin{proof}
It is clear that $\underline{x}\sim_{\Aut(H_n)} \underline{y}$ iff $\underline{x}'\sim_{S_n} \underline{y}'$ where 
$\underline{x}'=(0,x_2-x_1,\dots, x_{k}-x_1)$ and $\underline{y}'=(0,y_2-y_1,\dots,
y_{k}-y_1)$. Then $\W(\underline{x}')=\W(\underline{y}')$ iff $\mathcal{W}(\underline{x}')=\mathcal{W}(\underline{y}')$ and is left unchanged
if the coordinates are permuted. Conversely, for an appropriate
permutation $\sigma$ of the coordinates, $\sigma(\underline{x}')$ has its columns
reordered in lexicographic order. Another permutation $\tau$ has the
same effect on $\underline{y}'$; since $\mathcal{W}(\sigma(\underline{x}'))=\mathcal{W}(\tau(\underline{y}'))$, it means that
$\sigma(\underline{x}')=\tau(\underline{y}')$.
\end{proof}

\begin{remark}

\noindent

\begin{enumerate}
\item If $k=2$, we have of course $n_1(\ux)=d(x_1,x_2)$ and
  $n_0(\ux)=n-n_1(\ux)$.
In the case $k=3$, we have
\begin{align*}
n_{10}+n_{11}&=d(x_1,x_2)\\
n_{01}+n_{10}&=d(x_2,x_3)\\
n_{01}+n_{11}&=d(x_3,x_1)
\end{align*}
and the triple $(d(x_1,x_2),d(x_2,x_3),d(x_3,x_1))$
uniquely determines the orbit of $(x_1,x_2,x_3)$.
\item For arbitrary $k$,
taking into account the relation $\sum_w n_w=n$,  the orbits
  of $\Aut(H_n)$ on $H_n^{k}$ are described  by $2^{k-1}-1$
  independent parameters. In contrast, the orbits of $k$-tuples of
  elements of the unit sphere of the Euclidean space $S^{n-1}$ under
  the action of the orthogonal group $O(\R^n)$ need only
  $\binom{k}{2}$ real numbers in order to be uniquely determined,
  namely the pairwise inner products of the $k$ vectors. The orbits of $H_n^k$ under $\Aut(H_n)$ are determined by
the pairwise distances $d(x_i,x_j)$ only if  $k=2,3$.

\item In the next section we introduce several functions 
$f(x_1,\dots,x_k)$ such that
$$f(\sigma(x_1),\dots,\sigma(x_k))=f(x_1,\dots,x_k)$$ for all $\sigma\in
\Aut(H_n)$. It follows from the above description of the orbits of
$H_n^k$ that such  functions have an expression of the form $f(x_1,\dots,x_k)=\tilde
f(\W(\underline{x}))$. 
\end{enumerate}
\end{remark}

\section{$Aut(H_n)$-invariant functions on $H_n^{k}$.}\label{functions}

\subsection{The generalized Hamming distance.}

\begin{definition} The generalized
Hamming distance of $k$ elements of $H_n$ is defined by:
\begin{align*}
d(x_1,\dots,x_k)&=\card\{j,\  1\leq j\leq n:\  \card\{(x_1)_j,\dots,
(x_k)_j\}\geq 2\}\\
&=\card\{j,\  1\leq j\leq n:\  ((x_1)_j,\dots,(x_k)_j)\notin \{0^k,
1^k\}\}
\end{align*}
\end{definition}

\begin{proposition}\label{prop d}
The following properties hold for the generalized Hamming distance:
\begin{enumerate}
\item $d(x_1,x_2)$ is the usual Hamming distance.
\item For all permutation $\tau$ of $\{1,\dots,k\}$, $d(x_1,\dots,
  x_k)=d(x_{\tau(1)},\dots,x_{\tau(k)})$.
\item For all $\sigma\in \Aut(H_n)$,
  $d(x_1,\dots,x_k)=d(\sigma(x_1),\dots,\sigma(x_k))$.
The generalized distance $d(x_1,\dots,x_k)$ is related to
the weight distribution by: 
\begin{equation*}
d(x_1,\dots,x_k)=\sum_{w\neq 0^{k-1}}n_w(\underline{x}).
\end{equation*}
\item $d(x_1,\dots,x_{k-1},x_k)=d(x_1,\dots,x_{k-1})$ if $x_k$ belongs
  to the affine subspace generated by $x_1,\dots,x_{k-1}$.
\item ``Triangular'' inequality: for all $y\in H_n$, 
\begin{equation*} 
d(x_1,\dots,x_k)\leq \frac{1}{k-1}\sum_{i=1}^k d(x_1,\dots,
x_{i-1},y,x_{i+1},\dots, x_k).
\end{equation*}
\item The distance of three points can be expressed in terms of
  pairwise Hamming distances:
\begin{equation*}
d(x_1,x_2,x_3)=\frac{1}{2}(d(x_1,x_2)+d(x_2,x_3)+d(x_3,x_1)).
\end{equation*}
\item For more than three points we have only the inequality:
\begin{equation*}
d(x_1,x_2,\dots, x_k)\leq \frac{1}{k-1}\sum_{1\leq i<j\leq k}d(x_i,x_j).
\end{equation*}

\item We also have the inequalities:
\begin{align*}
d(x_1,\dots,x_k) &\leq 
\frac{1}{k-1}\sum_{i=1}^k d(x_1,\dots,x_{i-1},x_{i+1},\dots,x_k)\\
d(x_1,\dots,x_k) &\geq \frac{1}{k}\sum_{i=1}^k d(x_1,\dots,x_{i-1},x_{i+1},\dots,x_k).
\end{align*}

\end{enumerate}
\end{proposition}

\begin{proof}
Properties (i), (ii), (iii) are obvious. 

If 
$x_k$ belongs  to the affine subspace generated by
$x_1,\dots,x_{k-1}$, then we can write $x_k=\sum_{i=1}^{k-1} \alpha_i x_i$ with
$\sum_{i=1}^{k-1} \alpha_i =1$. Consequentely,
if $((x_1)_j,\dots,(x_{k-1})_j)= 0^{k-1}$, respectively $1^{k-1}$, then 
we have $((x_1)_j,\dots,(x_{k})_j)= 0^{k}$, respectively $1^{k}$. It
follows (iv) that $d(x_1,\dots,x_{k-1},x_k)=d(x_1,\dots,x_{k-1})$.

The announced ``triangular'' inequality (v) is easily checked in the case
$n=1$.
The general case follows from the fact that
\begin{equation}\label{e1}
d(x_1,\dots,x_k)=\sum _{j=1}^n d((x_1)_j,\dots,(x_k)_j).
\end{equation}

Again because of \eqref{e1}, it is enough to prove (vi) (vii) and (viii) for
$n=1$.
In this case, let the Hamming weight of $(x_1,\dots,x_k)$ be denoted
by $w$, then 
\begin{equation*}
    d(x_1,\dots,x_k)=\left\{
       \begin{array}{ll}
            1 &\text{ if } 1\leq w\leq k-1\\
            0 &\text{ if } w=0,k.
       \end{array}
    \right.
\end{equation*}
and 
\begin{equation*}
   \sum_{1\leq i<j\leq k}d(x_i,x_j)= w(k-w).
\end{equation*}
Obviously $w(k-w)\geq k-1$ if $w\neq 0,k$ 
and equals $0$ otherwise. Inequality (vii) follows. In the case
$k=3$, $w(k-w)$ takes only the values $0$ and $2$ hence (vi).

To prove (viii), notice that we have
\begin{equation*}
    \sum_{i=1}^kd(x_1,\dots,x_{i-1},x_{i+1},\dots,x_k)=\left\{
       \begin{array}{ll}
            k &\text{ if } 2\leq w\leq k-2\\
            k-1 &\text{ if } w=1,k-1\\
            0 &\text{ if } w=0, k
       \end{array}
    \right.
\end{equation*}
hence the announced inequalities.
\end{proof}

\subsection{The radial distance.}

\begin{definition}
The radial distance or radius of $k$ elements of $H_n$ is defined by:
\begin{align*}
r(x_1,\dots,x_k)&=\min\{r : \text{ there exists }y\in H_n \text{ s.t. }
  \{x_1,\dots,x_k\}\subset B(y,r)\}\\
&=\min_y \{\max_{1\leq i\leq k}  d(y,x_i) \}.
\end{align*}
\end{definition}

\begin{proposition}\label{prop r} The radial distance has the properties:

\begin{enumerate}
\item $r(x_1,x_2)= \lceil \frac{d(x_1,x_2)}{2}\rceil$.
\item For all permutations $\tau$ of $\{1,\dots,k\}$, $r(x_1,\dots,
  x_k)=r(x_{\tau(1)},\dots,x_{\tau(k)})$.
\item For all $\sigma\in \Aut(H_n)$,
  $r(x_1,\dots,x_k)=r(\sigma(x_1),\dots,\sigma(x_k))$.
\item For all $k$,
\begin{equation*}
r(x_1,\dots,x_k)\geq \max_{1\leq i\leq k} r(x_1,\dots,x_{i-1}, x_{i+1},\dots, x_k).
\end{equation*}
\item For $k=3$, we have
\begin{equation*}
r(x_1,x_2,x_3)=\max\{r(x_1,x_2), r(x_2,x_3), r(x_3,x_1)\}.
\end{equation*}
\end{enumerate}
\end{proposition}

\begin{proof} Properties (i) (ii) (iii) and (iv) are obvious.

Let $(x_1,x_2,x_3)\in H_n^3$. Without loss of generality we can assume
that $d(x_2,x_3)\leq d(x_3,x_1)\leq d(x_1,x_2)$ and that $x_1=0$. With 
the notation of section 1 it amounts to assume 
that $n_{01}\leq n_{10}\leq n_{11}$. 
Let $y\in H_n$ be the center of a smallest ball containing the three
words; 
clearly the coordinates of $y$ at the positions corresponding to
$w={00}$ in $M(\ux)$ must be equal
to $0$. Let $y_w$ be the number of ones  at the positions corresponding
to $w$. We have:
\begin{align*}
d(y,x_1)&= y_{01}+y_{10}+y_{11}\\
d(y,x_2)&= y_{01}+n_{10}-y_{10}+n_{11}-y_{11}\\
d(y,x_3)&= n_{01}-y_{01}+y_{10}+n_{11}-y_{11}\\
\end{align*}
We choose $y$ such that:
\begin{align*}
y_{01}&=0\\
y_{11}&=\lfloor \frac{n_{01}+n_{10}+2n_{11}}{4}\rceil\leq n_{11}\\
y_{10}&=\lfloor \frac{n_{10}-n_{01}}{4}\rceil\leq n_{10}
\end{align*}
Then one easily verifies that for $i=1,2,3$, $d(y,x_i)\leq \lceil\frac{n_{10}+n_{11}}{2}\rceil$
thus the ball centered at $y$ with radius
$\lceil\frac{n_{10}+n_{11}}{2}\rceil$ contains the three words
$x_1,x_2,x_3$. Since
$$n_{10}+n_{11}=d(x_1,x_2)=\max(d(x_2,x_3),d(x_3,x_1), d(x_1,x_2))$$
we have proved that 
\begin{equation*}
r(x_1,x_2,x_3)\leq \max\{r(x_1,x_2), r(x_2,x_3), r(x_3,x_1)\}.
\end{equation*}

\end{proof}

\begin{remark}
For $k\geq 4$ we cannot give  a nice expression of $r(\underline{x})$
as an explicit function of $\W(\underline{x})$.
It should be noted that the determination of the center $y$ and thus of $r(\ux)$ cannot
be performed by a sequence of local decisions at each coordinate or
even at each subset of coordinates corresponding to each $u$; in other
words property \eqref{e1} of $d(\;\;\; )$ does not hold  for $r$ and it
makes it more difficult to study. 
However for $k$ randomly chosen points, the distances of each point to
the center $y$ of the smallest ball containing them are expected to
have about the same value,
in other words the points are expected to be close to the border of the ball. When it is the
case, the radius of the $k$ points is approximated by a much nicer function, 
 called the average radius (or  moment of inertia), 
introduced in \cite{Bli1}. 
\end{remark}

\subsection{The average radial distance.}

\begin{definition}
The average radial distance or average radius (or moment distance or moment of inertia) of $k$ elements of $H_n$ is defined by:
\begin{align*}
\rr(x_1,\dots,x_k)&=\min_y \frac{1}{k}\sum_{1\leq i\leq k}  d(y,x_i).
\end{align*}
\end{definition}

\begin{proposition}\label{prop m} The average radius has the properties:
\begin{enumerate}
\item $\rr(x_1,x_2)= \frac{d(x_1,x_2)}{2}$.
\item For all permutation $\tau$ of $\{1,\dots,k\}$, $\rr(x_1,\dots,
  x_k)=\rr(x_{\tau(1)},\dots,x_{\tau(k)})$.
\item For all $\sigma\in \Aut(H_n)$,
  $\rr(x_1,\dots,x_k)=\rr(\sigma(x_1),\dots,\sigma(x_k))$. In terms of the
  weight distribution $\W(\ux)=(n_w(\ux))_{w\in \F_2^{k-1}}$,
\begin{equation*}
\rr(x_1,\dots, x_k)=\frac{1}{k}\sum_{w\in \F_2^{k-1}} \min(wt(w), k-wt(w))n_w(\ux)
\end{equation*}
\item For all $k$,
\begin{equation*}
\rr(x_1,\dots,x_k)\geq \frac{1}{k}\sum_{i=1}^k \rr(x_1,\dots,x_{i-1}, x_{i+1},\dots, x_k) 
\end{equation*}
\item The above inequality is an equality for $k=1\mod 2$. In
  particular, we have
\begin{equation*}
\rr(x_1,x_2,x_3)=\frac{\rr(x_1,x_2)+\rr(x_2,x_3)+\rr(x_3,x_1)}{3}
\end{equation*}
\item For all $k$,
\begin{equation*}
\rr(x_1,\dots,x_k)\leq \frac{2(k-1)}{k(k-2)}\sum_{i=1}^k \rr(x_1,\dots,x_{i-1}, x_{i+1},\dots, x_k).
\end{equation*}
\item ``Triangular'' inequality: for all $y\in H_n$,
\begin{equation*}
\rr(x_1,\dots,x_k)\leq \frac{1}{k-1}\sum_{i=1}^k
\rr(x_1,\dots,x_{i-1},y,x_{i+1},\dots,x_k).
\end{equation*}
\end{enumerate}
\end{proposition}

\begin{proof}
Properties (i) (ii) and the $\Aut(H_n)$-invariance are trivial. 

If the $j$-th column of the matrix $M(\ux)$ equals $u\in \F_2^{k}$,
the contribution of this column in $\sum_i d(y,x_i)$ is 
equal to $wt(u)$ if $y_j=0$ and to $wt(\overline{u})$ if $y_j=1$. So
the minimum of this sum over all $y$ equals
\begin{equation*}
\sum_u \min(wt(u),wt(\overline{u}))n_u(\ux).
\end{equation*}
which leads to the formula announced in (iii).
It also shows that 
\begin{equation*}
\rr(x_1,\dots,x_k)=\sum_{j=1}^n \rr((x_1)_j,\dots, (x_k)_j).
\end{equation*}
Consequently, in order to prove  the remaining assertions, we can assume $n=1$.
Let the weight of $\ux$ be denoted by $w$. Without
loss of generality we assume that either $w< k/2$ or 
$w=k/2$. This last case can only happen if
$k=0\mod 2$. We prove (v) and (vi): in the case $w<k/2$, removing $x_i=1$ 
makes $k\rr(\ux)$ drop by $1$ while removing $x_i=0$
does not change $k\rr(\ux)$. In the case $w=k/2$, $k\rr(\ux)$
always drops by $1$.
In other words,
\begin{equation*}
(k-1)\rr(x_1,\dots, x_{i-1},x_{i+1},\dots, x_k)=\left\{\begin{array}{ll}
w -1 &\text{ if }x_i=1\text{ and }w<\frac{k}{2}\\
w  &\text{ if }x_i=0\text{ and }w<\frac{k}{2}\\
w -1 &\text{ if } w=\frac{k}{2}
\end{array}\right.
\end{equation*}
and
\begin{equation*}
(k-1)\sum_{i=1}^k \rr(x_1,\dots, x_{i-1},x_{i+1},\dots, x_k)=\left\{\begin{array}{ll}
(k-1) w &\text{ if }w<\frac{k}{2}\\
k(k-2)/2 &\text{ if } w=\frac{k}{2}.
\end{array}\right.
\end{equation*}
We obtain 
\begin{equation*}
\frac{1}{k}\sum_{i=1}^k \rr(x_1,\dots, x_{i-1},x_{i+1},\dots, x_k)=\left\{\begin{array}{ll}
\rr(\ux) &\text{ if } w\neq \frac{k}{2}\\
\frac{(k-2)}{(2k-2)}\rr(\ux) &\text{ if } w=\frac{k}{2}
\end{array}\right.
\end{equation*}
hence the inequalities 
\begin{equation*}
\frac{(k-2)}{(2k-2)}\rr(\ux)\leq \frac{1}{k}\sum_{i=1}^k \rr(x_1,\dots, x_{i-1},x_{i+1},\dots, x_k)\leq \rr(\ux).
\end{equation*}
If $k=1\mod 2$,  the case $w=k/2$ never happens so the second
inequality is always an equality.

For the triangular inequality (vii), we find
\begin{align*}
k\sum_{i=1}^k \rr(x_1,\dots, x_{i-1},&y,x_{i+1},\dots, x_k)=\\
&\left\{\begin{array}{ll}
(k-1)w+k &\text{ if } w<\lfloor \frac{k}{2}\rfloor  \text{ and } y=1\\
(k-1)w &\text{ if } w=\frac{k}{2} \text{ and } y=1\\
kw &\text{ if } w=\frac{k-1}{2} \text{ and } y=1\\
(k-1)w &\text { if } y=0
\end{array}\right.
\end{align*}
hence 
\begin{equation*}
k\sum_{i=1}^k \rr(x_1,\dots, x_{i-1},y,x_{i+1},\dots, x_k)\geq (k-1)w=
k(k-1)\rr(\ux).
\end{equation*}
\end{proof}

\subsection{Relationships between $d$, $r$, $\rr$}

\begin{proposition}\label{prop compare}
The following hold:
\begin{enumerate}
\item For all $\ux=(x_1,\dots,x_k)\in H_n^{k}$,
\begin{equation*}
\frac{1}{k} d(\ux)\leq \rr(\ux)\leq r(\ux)\leq d(\ux)
\end{equation*}
and 
\begin{equation*}
\rr(\ux)\leq \frac{1}{2} d(\ux).
\end{equation*}
\item For $k=2,3$, $d(\ux)=k\rr(\ux)$.
\item If $\rr(\ux)=r(\ux)$ then the center of any of the balls of minimal radius
  $r(\ux)$
containing the points $(x_1,\dots,x_k)$ is equidistant to these points.
  The converse is false, in the sense that the points may be
  equidistant to some $y$ while $\rr(\ux)<r(\ux)$.
\end{enumerate}
\end{proposition}

\begin{proof}
Since 
\begin{equation*}
\frac{1}{k}\sum_{i=1}^k d(y,x_i)\leq \max_i d(y,x_i),
\end{equation*}
we obviously have $\rr(\ux)\leq r(\ux)$. 
From
\begin{equation*}
1\leq \min(wt(w),k-wt(w))\leq k/2
\end{equation*}
for $w\neq 0^{k-1}$ and from the expressions given in
Proposition \ref{prop m} (iii) for $\rr(\ux)$ and in 
Proposition~\ref{prop d} (iii) for $d(\ux)$ we have
\begin{equation*}
\frac{1}{k}d(\ux) \leq \rr(\ux) \leq \frac{1}{2} d(\ux).
\end{equation*}
Let $J$ be the set of coordinates where $\ux_j\in \{0^k,1^k\}$. If we
choose $y$ such that $y_j$ agrees with $(x_i)_j$ for $j\in J$, then 
$d(y,x_i)\leq n-|J|=d(\ux)$. Thus $r(\ux)\leq d(\ux)$. This concludes
point (i).
 
(ii) is obvious from previous formulas.

Let us assume that $\rr(\ux)=r(\ux)=r$ and let $y$ be the center of a
ball of radius $r$ containing all $x_i$. Then we have the inequalities
\begin{equation*}
r=\rr(\ux)\leq \frac{1}{k}\sum_i d(y,x_i)\leq \max_i d(y,x_i)=r
\end{equation*}
thus $\frac{1}{k}\sum_i d(y,x_i)=\max_i d(y,x_i)$ which means
that all $d(y,x_i)$ are equal to $r$. 

We build a counterexample with $k=3$. If $n_{01}, n_{10}, n_{11}$ are
even numbers, the points will be equidistant
to some point $y$ with $y_w=n_w/2$. We assume moreover that 
$n_{01}\leq n_{10}\leq n_{11}$. From Proposition \ref{prop r},
$r(\ux)=(n_{10}+n_{11})/2$ and from Proposition \ref{prop m},
$\rr(\ux)=(n_{01}+n_{10}+n_{11})/3$ so if $2n_{01}<n_{10}+n_{11}$ we are done.
\end{proof}

\section{Code invariants and their significance}\label{invariants}

\subsection{Code invariants}

\begin{definition}\label{def}
For any $C\subset H_n$, and for $f=d,r,\rr$, we define
\begin{equation*}
f_{k-1}(C)=\min f(x_1,\dots, x_k)
\end{equation*}
where the minimum is taken over all $k$-tuples of pairwise distinct 
elements of $C$. Moreover we define
\begin{equation*}
d_{k-1}^{\text{\em aff}}(C)=\min d(x_1,\dots, x_k)
\end{equation*}
where the minimum is taken over all $k$-tuples of affinely independent 
elements of $C$. Following standard notation in coding theory, 
we let $A_{k-1}(n,f,m)$ be the maximal number of elements that a code
$C\subset H_n$ can have under the condition $f_{k-1}(C)\geq m$.
\end{definition}

% For the generalized Hamming distance $d$, the 
% first definition seems more adequate for general codes while the
% second one, introduced in \cite{CLZ}, is better suited for linear (or
% affine) codes.

\begin{proposition}
The following hold:
\begin{enumerate}
\item $d_1(C)=d_1^{\text{\em aff}}(C)$ is the Hamming distance of the code $C$.
\item $d_2(C)=d_2^{\text{\em aff}}(C)$.
\item If $C$ is a linear code, and
$2^{t-1}< k\leq 2^t$, $d_{k-1}(C)=d_{t}^{\text{\em aff}}(C)$. 
\item If $C$ is a linear code, $\dkaff(C)$ coincides with the minimum $(k-1)$-th
  generalized weight as defined in \cite{W}, namely:
\begin{equation*}
\dkaff(C)=\min \{ w(D) : D\subset C, D\  linear, \dim(D)=k-1\}.
\end{equation*}
where $w(D)$ is the set of coordinates $i$ at which at least one
element of $D$ is non zero.
\end{enumerate}
\end{proposition}

\begin{proof}
Obvious.
\end{proof}

\noindent
{\bf Remark.} The quantity $d_k(C)$ is more natural and
easier to deal with than the more intricate $d_k^{\text{aff}}(C)$.
Unfortunately, $d_k(C)$ only coincides with the minimum $k$-th generalized
weight of a linear code for $k=1,2$, hence the definition 
of $d_k^{\text{aff}}(C)$, originally stated in \cite{CLZ}.
In \cite{Bas} yet another generalisation of the minimum $k$-th generalized
weight to non-linear codes is introduced that does not consider
affinely independent sets of vectors. We will not dwell on the
differences here and our study will
mostly focus on the quantity $d_k(C)$ itself, of interest in its own
right since it has a natural interpretation in terms of list decoding
``radius'' for lists of size $k$ when decoding from erasures (see
section \ref{sec:list} below).

\begin{proposition}
For $f=d,\daff,r,\rr$ and for any code $C$,
\begin{equation*}
f_{k-1}(C)\leq f_{k}(C)
\end{equation*}
\end{proposition}
\begin{proof} It follows from Propositions 
\ref{prop d} (viii), \ref{prop r} (iv) and \ref{prop m} (iv) that for
pairwise distinct $\ux$ 
\begin{equation*}
f(x_1,\dots,x_{k+1})\geq f_{k-1}(C)
\end{equation*}
respectively for affinely independent $\ux$
\begin{equation*}
d(x_1,\dots,x_{k+1})\geq d_{k-1}^{\text{aff}}(C).
\end{equation*}
Hence $f_{k-1}(C)\leq f_k(C)$.
\end{proof}

\subsection{List decoding}\label{sec:list}

A list decoding procedure is a decoding procedure that outputs a list
of codewords. The length $L$ of the list is determined in advance. 
This list is usually obtained by the enumeration of all codewords 
in a ball $B(y,r)$. For a given code $C$, the associated value of $r$ 
is known as the $L$-list decoding radius of $C$:

\begin{definition}
The $L$-list decoding radius $R_L(C)$ is the largest value of $r$ such
that,
for all $y\in H_n$, 
\begin{equation*}
\card(B(y,r)\cap C)\leq L.
\end{equation*}
\end{definition}
In the case $L=1$, we recover the notion of the (unique) decoding radius of a code,
$R_1(C)=\lfloor(d(C)-1)/2\rfloor$. This number is also the largest
value of $r$ such that the balls of radius $r$ centered at the
codewords have the property that any $L+1$ of them have an empty
intersection.
A set  of balls with this property is called a $L$-multiple packing.
Thus a classical packing of balls is a $1$-multiple packing.

\begin{proposition}\label{prop list radius}
\begin{equation*}
R_L(C)=r_{L}(C)-1.
\end{equation*}
\end{proposition}
\begin{proof}
There exists $(x_1,\dots,x_{L+1})\in C^{L+1}$ and $y\in H_n$ such that
for all $1\leq i\leq L+1$, $x_i\in B(y,r_{L}(C))$ and $x_i\neq x_j$  thus
$\card(B(y,r_{L}(C))\cap C)=L+1$ and $R_L(C)< r_L(C)$. Moreover,
if $r<r_{L}(C)$, $L+1$ codewords cannot be elements of the same ball of
radius $r$ thus $R_L(C)=r_{L}(C)-1$.
\end{proof}

The notion of list decoding can also be investigated in the framework of
erasure decoding, see \cite{Gu2}. 

\begin{definition}
The $L$-list decoding radius for erasures $R^{\text{\em er}}_L(C)$ is
the largest value of $r$ such that,
for all $E\subset \{1,\dots,n\}$, $\card(E)\leq r$,
and for any $y=(y_i)_{i\not\in E}\in\{0,1\}^{n-\card(E)}$
 \begin{equation*}
\card(\{x\in C : (x_i)_{i\not\in E}=y\})\leq L.
\end{equation*}
\end{definition}

The following proposition, which is a straightforward consequence of
the definition of $d_{L}$, makes generalized distances relevant to
erasure decoding \cite{Gu2,Z,ZC}.

\begin{proposition}
\begin{equation*}
R_L^{\text{\em er}}(C)=d_{L}(C)-1.
\end{equation*}
\end{proposition}
% However the notion of generalized Hamming distance was  not introduced 
% for this purpose but rather in view of cryptographic applications
% (\cite{W}).

\section{Upper bounds for $d_k$,  $r_k$,  $\rr_k$.}\label{bounds}

In this section we gather the analogs of the Singleton, Hamming,
Plotkin and Elias bounds for $f=d,r,\rr$. 
The methods are well-known and some of the bounds may be found
explicitely in the litterature, but not always in form precise enough
for numerical computation (in particular only asymptotic versions of
the Elias bounds can be found) which we need to compare them to the
new SDP bounds.

\subsection{The Singleton bound} This bound for $d$ is the most
elementary and is a natural generalisation of the classical Singleton
bound for the ordinary Hamming distance.

\begin{proposition}
Let $C\subset H_n$. Then, 
if $d_{k-1}(C)\geq d_{k-1}$ 

\begin{equation*}
|C|\leq (k-1)2^{n-d_{k-1}+1}.
\end{equation*}
\end{proposition}

\begin{proof}
Consider the restriction of
the codewords on a fixed set of $(n-d_{k-1}+1)$ indices. The number of
possible images is of course $2^{n-d_{k-1}+1}$. If 
$|C|>(k-1)2^{n-d_{k-1}+1}$, there is a subset of $k$ codewords having
the same image. Thus they have a generalized Hamming distance at most
equal to $d_{k-1}+1$ and we have a contradiction.
\end{proof}

It is worth noticing that the Singleton bound for $k=3$ is tight
for $d=3$ and for $d=n$.

\subsection{Hamming type bound}

This volume type bound is established in \cite{CLZ}[Prop II.I] for
$d_{k-1}$ and for linear codes and generalized to the non-linear case
in \cite{Bas}. 
We take the following notations: the number of elements of a ball of
radius $r$ in $H_n$ is denoted $b^n_r$ or $b_r$ if $n$ is clear from
the context.
We recall the formula
\begin{equation*}
b_r^n =\sum_{k=0}^r \binom{n}{k}.
\end{equation*}

\begin{proposition}
Let $C\subset H_n$. Then 
\begin{enumerate}
\item If $r_{k-1}(C)\geq r_{k-1}$ or $\rr_{k-1}(C)\geq r_{k-1}$ then 
\begin{equation*}
|C|\leq \frac{(k-1)2^{n}}{b_{r_{k-1}-1}^n}
\end{equation*}

\item If $d_{k-1}(C)\geq d_{k-1}$ then 

\begin{equation*}
|C|\leq \frac{(k-1)2^n}{b^n_{\lceil {d_{k-1}}/{k}\rceil-1}}
\end{equation*}

\item If $\dkaff(C)\geq d_{k-1}$ then 
\begin{equation*}
|C|\leq \frac{2^{n+k-2}}{b^n_{\lceil {d_{k-1}}/{k}\rceil-1}}
\end{equation*}
\end{enumerate}
\end{proposition}

\begin{proof}
(i) If $r_{k-1}(C)\geq r_{k-1}$ or $\rr_{k-1}(C)\geq r_{k-1}$, from Proposition
  \ref{prop compare} (i) and Proposition \ref{prop list radius} we have, for all $y\in H_n$, $\card(B(y,r_{k-1}-1)\cap C)\leq k-1$.
In order to establish the announced inequality, we count in two ways
the elements of 
\begin{equation*}
E:=\{ (c,y), c \in C, y \in H_n : d(c,y)\leq r_{k-1}-1\}.
\end{equation*}
We have 
\begin{align*}
\card(E)&=\sum_{c\in C} \card\{ y\in H_n  : d(y,c)\leq r_{k-1}-1 \}\\
&=|C|b_{r_{k-1}-1}^n\\
&=\sum_{y\in H_n} \card\{ c\in C  : d(y,c)\leq r_{k-1}-1\}\\
&\leq \card(H_n) (k-1)=(k-1)2^n.
\end{align*}

(ii) 
If $d_{k-1}(C)\geq d_{k-1}$, from Proposition \ref{prop compare} (i) we have 
$r_{k-1}(C)\geq \lceil \frac{d_{k-1}}{k}\rceil $ thus we can apply 
the previous result.

(iii) 
Let $(x_1,\dots,x_k)\in C^k$ be affinely independent and let $y\in H_n$. We have
\begin{align*}
d_{k-1}\leq d(x_1,\dots,x_k) &\leq \frac{1}{k-1}\sum_{1\leq i<j\leq k}
d(x_i,x_j)\\
& \leq \frac{1}{k-1}\sum_{1\leq i<j\leq k} (d(x_i,y)+d(y,x_j))\\
& \leq \sum_{1\leq i\leq k} d(x_i,y).
\end{align*}
Thus for some $i$, $1\leq i\leq k$, $d(x_i,y)\geq \lceil \frac{d_{k-1}}{k}\rceil$.
Since any subset of $H_n$ with at least $2^{k-2}+1$ elements contains $k$ affinely
independent ones,  we have  for all $y\in H_n$, 
\begin{equation*}
\card(B(y, \lceil \frac{d_{k-1}}{k}\rceil-1) \cap C) \leq  2^{k-2}.
\end{equation*}
and we follow the same line as in (i).
\end{proof}

\subsection{Plotkin type bound}
This type of bound is usually derived from the estimate of the average value of
$f$ among $C^{k}$. This average value can be estimated when $f$ can be calculated
from its value at each coordinate, which is the case for $f=d,\rr$. 

We take the following notations: let $C$ be a binary code with $M$
elements; let $w_j$ be the number of ones in the $j$-th column of the $M\times n$ matrix whose rows are the $M$
  elements of $C$. Let $J_k(C)$,
  respectively $J_k^{\operatorname{aff}}(C)$  be the set of
$k$-tuples of pairwise distinct, respectively affinely independent
codewords. We moreover define 
\begin{equation*}
j_k(x):=\begin{cases}
0 &\text{ if } x\leq k-1 \\
\prod_{t=0}^{k-1}(x-t) &\text{ if } x\geq k-1
\end{cases}
\end{equation*}
and
\begin{equation*}
j_k^{\operatorname{aff}}(x):=\begin{cases}
0 &\text{ if } x\leq 2^{k-2} \\
x\prod_{t=0}^{k-2}(x-2^t) &\text{ if } x\geq 2^{k-2}.
\end{cases}
\end{equation*}
We have obviously $|J_k(C)|=j_k(M)$ and 
$|J_k^{\operatorname{aff}}(C)|\geq j_k^{\operatorname{aff}}(M)$, this last
inequality being an equality if $C$ is linear. 
For $x\in \R$, we also denote as is usual $\binom{x}{k}:=j_k(x)/k!$.

\begin{proposition}\label{prop Plotkin} With the above notations:
\begin{enumerate}
\item If $d_{k-1}(C)\geq d_{k-1}$ then 
\begin{equation*}
\delta_{k-1}:=\frac{d_{k-1}}{n}\leq 1-2\frac{\binom{M/2}{k}}{\binom{M}{k}}.
\end{equation*}
\item If $C$ is linear or if $k=3$, and if $\dkaff(C)\geq d_{k-1}$, we have 
\begin{equation*}
\delta_{k-1}:=\frac{d_{k-1}}{n}\leq \big(1-\frac{1}{2^{k-1}}\big)\frac{M}{M-1}.
\end{equation*}
\item If $\rr_{k-1}(C)\geq r_{k-1}$ then 
\begin{equation*}\label{b2}
\rho_{k-1}:=\frac{r_{k-1}}{n}\leq \frac{\sum_{i=1}^{k-1} \frac{1}{k} \binom{M/2}{i}\binom{M/2}{k-i}\min(i,k-i)}{\binom{M}{k}}
\end{equation*}
\end{enumerate}
\end{proposition}

\begin{proof}
(i) For the generalized Hamming distance, we have
\begin{align*}
 \sum_{\underline{x}\in J_k(C)}
d(\underline{x})
&= \sum_{\underline{x}\in J_k(C)}\big(\sum_{j=1}^n d((x_{1})_j,\dots, (x_k)_j)\big)\\
&= \sum_{j=1}^n\big(\sum_{\underline{x}\in J_k(C)} d((x_{1})_j,\dots, (x_k)_j)\big)\\
&= \sum_{j=1}^n \sum_{i=1}^{k-1} k! \binom{w_j}{i}\binom{M-w_j}{k-i}.
\end{align*}
The function $x\to
\binom{x}{i}\binom{M-x}{k-i}+\binom{x}{k-i}\binom{M-x}{i}$ is concave
and invariant by $x\to M-x$ thus it takes its maximum at $x=M/2$. We
derive the inequalities:
\begin{align*}
j_k(M) d_{k-1}\leq d(\underline{x})\leq n \sum_{i=1}^{k-1} k! \binom{M/2}{i}\binom{M/2}{k-i}=nk!\Big(\binom{M}{k}-2\binom{M/2}{k}\Big).
\end{align*}

(ii) In the special case $k=3$, we obtain from (i) the desired
inequality. In the case $C$ linear, we observe that
$w_j=0,M,M/2$ and that $d((x_{1})_j,\dots, (x_k)_j)$ is non zero only
if $w_j=M/2$ and $x_{1},\dots, x_k$ do not all belong to $\{x\in C:
x_j=0\}$ or to $\{x\in C: x_j=1\}$, which have $M/2$ elements. Thus
\begin{align*}
j_k^{\operatorname{aff}}(M) d_{k-1}\leq n\big(j_k^{\operatorname{aff}}(M) -2j_k^{\operatorname{aff}}(M/2) \big)
\end{align*}
hence the announced inequality.

(iii) The result for $\rr$ is derived similarly to the result (i) in
the $d$ case.
\end{proof}

\begin{remark} The upper bounds established in Proposition \ref{prop
    Plotkin} can easily be turned into upper bounds for $M=|C|$. Indeed,
  if $\phi_{k-1}:=f_{k-1}/n\leq A(M)/B(M)$ where $A$ and $B$ are polynomials of the
  same degree, with respective leading coefficients $\alpha$ and
  $\beta$, with $B(M)>0$, then, if $\phi_{k-1}\geq \alpha/\beta$, $M$ is upper bounded
  by the largest zero of the polynomial $\phi_{k-1}B(M)-A(M)$. The
  bound obtained this way holds for $\delta_{k-1}\geq 1-1/2^{k-1}$ and 
$\rho_{k-1}\geq 1/2-\binom{k-1}{\lfloor(k-1)/2\rfloor}/2^k$.
\end{remark}

\subsection{The Elias-Bassalygo technique and constant weight codes}

We recall  that $A_{k-1}(n,f,m)$ denotes the maximal number of elements that a code
$C\subset H_n$ can have under the condition $f_{k-1}(C)\geq m$;
analogously let $A_{k-1}(n,w,f,m)$ be the maximum among the  codes with constant
weight $w$. With a standard argument, the following inequality holds:
\begin{equation}\label{Elias}
\frac{A_{k-1}(n,f,m)}{\card(H_n)}\leq \frac{A_{k-1}(n,w,f,m)}{\card(J_{n}^w)}
\end{equation}
where $J_n^w$ is the set of the $\binom{n}{w}$ binary words of length $n$ and weight $w$.
This so-called {\em Elias  Bassalygo technique} is expected to improve the
bounds on $H_n$, if similar bounds are established on the Johnson
spaces $J_n^w$. Note that the value of $w$ on the right hand side can
be chosen freely. This line was followed in \cite{CLZ} for the
generalized Hamming distance, and required moreover to extend the
methods to non linear codes. In view of \eqref{Elias}, we work
out  Plotkin type bounds for constant weight codes:

\begin{proposition}
Let $C\subset J_n^w$ have $M$ elements and let $\omega:=w/n$.
\begin{enumerate}
\item If $d_{k-1}(C)\geq d_{k-1}$ then 
\begin{equation*}
\delta_{k-1}:=\frac{d_{k-1}}{n}\leq 1-\frac{\binom{M\omega}{k}+\binom{M(1-\omega)}{k}}{\binom{M}{k}}.
\end{equation*}
\item If $\dkaff(C)\geq d_{k-1}$, we have 
\begin{equation*}
\delta_{k-1}:=\frac{d_{k-1}}{n}\leq \frac{j_k(M)}{j_k^{\operatorname{aff}}(M)}\Big(1-\frac{\binom{M\omega}{k}+\binom{M(1-\omega)}{k}}{\binom{M}{k}}\Big).
\end{equation*}

\item If $\rr_{k-1}(C)\geq r_{k-1}$ then 
\begin{equation*}\label{bw2}
\rho_{k-1}:=\frac{r_{k-1}}{n}\leq \frac{\sum_{i=1}^{k-1} \frac{1}{k} \binom{M\omega}{i}\binom{M(1-\omega)}{k-i}\min(i,k-i)}{\binom{M}{k}}.
\end{equation*}
\end{enumerate}

\end{proposition}
\begin{proof} For $d$ and $\rr$, we follow the same line as for the
  proof
of Proposition \ref{prop Plotkin}. There we applied the inequality 
$\sum_{j=1}^ng(w_j)\leq ng(M/2)$ for relevant  functions $g$, being
concave and invariant by $x\to M-x$.
Since $C\subset J_n^w$, we have
$\sum_{j=1}^n w_j=Mw$, so we can instead use the stronger inequality
$\sum_{j=1}^n g(w_j)\leq ng(Mw/n)$.  
\end{proof}

\section{The SDP bound for $d_2$, $r_2$, $\rr_2$.}\label{SDP}

The method developed in \cite{S} can be used to derive upper
bounds for the cardinality of a binary code $C$ with given $d_2(C)$
(respectively $r_2(C)$, $\rr_2(C)$).
Recall that $d_2(C)\geq d$ if and only if, for all $(x,y,z)\in C^3$ 
such that $x\neq y$, $y\neq z$, 
$d(x,y)+d(y,z)+d(z,x)\geq 2d$ , $z\neq x$ (respectively 
$r_2(C)\geq r$ iff
$\max(\lceil \frac{d(x,y)}{2} \rceil, \lceil \frac{d(y,z)}{2}
\rceil, \lceil \frac{d(z,x)}{2} \rceil)\geq r$ and 
$\rr_2(C)\geq \rr$ iff
$d(x,y)+d(y,z)+d(z,x)\geq 6\rr$ for all $(x,y,z)\in C^3$).

The SDP constraints at work in \cite{S} are exactly SDP constraints on
triples of points. In order to describe these constraints
we adopt the group theoretic point of view of \cite{Ba}, \cite{V1},
\cite{V2}.

Let $X:=H_n$ and, for all $k:=0\dots n$, the so-called Johnson spaces $X_k:=\{x, x\in X : wt(x)=k\}$.
We consider the action of the symmetric group $S_n$ on $H_n$. The
Johnson spaces $X_k$ are exactly the orbits of this action. Now we
consider the decomposition of the functional space $L^2(X)=\R^{X}$ of
real valued functions on $X$
under the action of $S_n$. The space $\R^X$ is endowed
with the $S_n$-invariant scalar product

\begin{equation*}
(f, g)=\frac{1}{|X|} \sum_{x\in X} f(x)g(x).
\end{equation*}

We have the obvious decomposition into pairwise orthogonal $S_n$-invariant subspaces:
\begin{equation*}
\R^X=\R^{X_1}\perp \R^{X_1}\perp\dots\perp \R^{X_n}.
\end{equation*}
The decomposition of $\R^{X_k}$ into $S_n$-irreducible subspaces is
described in \cite{Del2}. We have 
\begin{equation*}
\R^{X_k}=H_{0,k}\perp H_{1,k}\perp\dots \perp H_{\min(k,n-k),k}
\end{equation*}
where the $H_{i,k}$ are pairwise  isomorphic for fixed $i$ and 
pairwise non  isomorphic for fixed $k$.
The picture looks like:

\begin{equation*}
\begin{array}{cccccccc}
\R^{X}=&\R^{X_1} \perp &\R^{X_1}\perp &\dots  &\perp \R^{X_{\lfloor
    \frac{n}{2}\rfloor}}\perp  &\dots &\perp \R^{X_{n-1}}&\perp
\R^{X_n}\\
&&&&&&&\\
       & H_{0,0} \perp &H_{0,1} \perp      &\dots  &\perp H_{0,\lfloor         \frac{n}{2}\rfloor}\perp  &\dots &\perp H_{0,n-1}&\perp H_{0,n}\\
       &               & H_{1,1}\perp       &\dots  &             &     & \perp H_{1,n-1}&\\
&&&\ddots &\vdots &&&\\
&&&& H_{\lfloor \frac{n}{2}\rfloor,\lfloor \frac{n}{2}\rfloor}
\end{array}
\end{equation*}
where the columns represent the decomposition of $\R^{X_k}$ and 
the rows  the isotypic components of $\R^X$, with multiplicity $n-2k+1$, i.e. 
we have for  $0\leq k\leq \lfloor \frac{n}{2}\rfloor$, 
$$H_{k,k}\perp H_{k,k+1}\perp\dots\perp H_{k,n-k}\simeq H_{k,k}^{n-2k+1}.$$

To each of these isotypic components, indexed by $k$, for $0\leq k\leq
\lfloor \frac{n}{2}\rfloor$,   we associate a matrix 
$E_k$ 
of size $n-2k+1$ as explained in \cite{BV1}, indexed with
$s,t$ subject to  $k\leq s,t\leq n-k$, in the following way:
Let $(e_{k,k,1},e_{k,k,2},\dots, e_{k,k,h_k})$ be an orthonormal basis of
$H_{k,k}$ and let $e_{k,s,j}=\psi_{k,s}(e_{k,k,j})$. The application
$\psi_{k,s}$ is defined by: 

\begin{equation*}
\begin{array}{llll}
\psi_{k,s}: & \R^{X_k} & \to &\R^{X_k}\\
& f &\mapsto  &\psi_{k,s}(f) : 
\psi_{k,s}(f)(y)=\sum_{\substack{wt(x)=k\\x\subset y}}f(x)
\end{array}
\end{equation*}

and has the property to send and orthonormal basis of
$H_{k,k}$ to an orthogonal basis of $H_{k,s}$, the elements of this basis
having constant square norm equal to
$\binom{n-2k}{s-k}$. The $(s,t)$ coefficient of $E_k$ is
defined by:

\begin{equation*}
E_{k,s,t}(x,y)=
\frac{1}{h_k}\sum_{j=1}^{h_k} e_{k,s,j}(x)e_{k,t,j}(y).
\end{equation*}

From \cite{BV1}, $E_{k,s,t}(x,y)=E_{k,s,t}(gx,gy)$ for all $g\in S_n$. 
Thus for $k\leq s \leq t\leq n-k $, we can define $P_{k,s,t}$ by
 $E_{k,s,t}(x,y)=P_{k,s,t}(s-|x\cap y|)$. It turns out that 
these
 $P_{k,s,t}$ 
express in terms of Hahn polynomials.

The Hahn polynomials associated to the parameters
$n,s,t$ with $0\leq s\leq t\leq n$  are the polynomials 
$Q_k(n,s,t; x)$ with $0\leq k\leq \min(s,n-t)$ uniquely determined by
the properties:
\begin{enumerate}
\item $Q_k$ has degree $k$ in the variable $x$
\item They are orthogonal polynomials for the weights
\begin{equation*}
0\leq i\leq s \quad w(n,s,t; i)=\binom{s}{i}\binom{n-s}{t-s+i}
\end{equation*}
\item $Q_k(0)=1$
\end{enumerate}

The combinatorial meaning of the above weights is the following:

\begin{lemma}
Given $x\in X_k$, the number of elements $y\in X_t$ such that $|x\cap
y|=s-i$ is equal to $w(n,s,t;i)$.
\end{lemma}

Finaly we have:

\begin{proposition}\label{p1} If $k\leq s\leq t\leq n-k$, $wt(x)=s$, $wt(y)=t$,
\begin{equation*} 
E_{k,s,t}(x,y)= |X| \frac{\binom{t-k}{s-k}\binom{n-2k}{t-k}}{\binom{n}{t}\binom{t}{s}}Q_k(n,s,t; s-|x\cap y|)
\end{equation*}
If $wt(x)\neq s$ or $wt(y)\neq t$, $E_{k,s,t}(x,y)=0$.
\end{proposition}

By the construction, the matrices $E_k$ satisfy  the semidefinite
positivity  properties:

\begin{theorem}
For all $k$, $0\leq k\leq \lfloor \frac{n}{2}\rfloor$, for all
$C\subset H_n$,
\begin{equation*}
\sum_{(c,c')\in C^2} E_k(c,c')\succeq 0.
\end{equation*}
\end{theorem}

These constraints are not interesting for pairs because they are not
stronger than the positivity properties from Delsarte method. They are
only interesting if triples of points are involved: namely we associate
to
$(x,y,z)\in H_n^3$ the matrices 

\begin{equation*}
F_k(x,y,z):=E_k(x-z,y-z).
\end{equation*}

We have for all
$C\subset H_n$, and for all $z\in H_n$,
\begin{equation*}
\sum_{(c,c')\in C^2} F_k(c,c',z)\succeq 0
\end{equation*}
which leads to the two positive semidefinite conditions:

\begin{equation}\label{eq1}
\left\{
\begin{array}{ll}
&\sum_{(c,c',c")\in C^3} F_k(c,c',c")\succeq 0\\
&\sum_{(c,c')\in C^2, c"\notin  C} F_k(c,c',c")\succeq 0
\end{array}
\right.
\end{equation}

From Proposition \ref{p1}, $E_k(x-z,y-z)$ only depends on the values
of $wt(x-z)$, $wt(y-z)$, $wt(x-y)$; so with 
$a:=d(y,z)$, $b:=d(x,z)$,
$c:=d(x,y)$, we have for some matrices $T_k(a,b,c)$,

\begin{equation*}
F_k(x,y,z)=T_k(a,b,c).
\end{equation*}

\medskip
We introduce the unknowns $x_{a,b,c}$ of the SDP.
Let, for
\begin{equation*}
(a,b,c)\in \Omega:=\left. \{(a,b,c) \in [0\dots n]^3 : 
\begin{array}{ll}
& a+b+c\equiv 0\mod 2\\
&a+b+c\leq 2n\\
&c\leq a+b\\
&b\leq a+c\\
&a\leq b+c
\end{array}\right\}
\end{equation*}

\begin{equation*}
x_{a,b,c}:=\frac{1}{|C|} \card\{(x,y,z)\in C^3:
d(y,z)=a, d(x,z)=b, d(x,y)=c\}.
\end{equation*}

Note that 
\begin{equation*}
x_{0,c,c}=\frac{1}{|C|} \card\{(x,y)\in C^3:
d(x,y)=c\}.
\end{equation*}

With the definition
\begin{equation*}
\begin{array}{ll}
t(a,b,c)&:=\card\{z \in H_n : d(x,z)=b \text{ and }d(y,z)=a\}\text{ for
} d(x,y)=c\\
&=\binom{c}{i}\binom{n-c}{a-i}
\text{ where } a-b+c=2i 
\end{array}
\end{equation*}
the following inequalities hold for $x_{a,b,c}$ :

\begin{enumerate}
\item $x_{0,0,0}=1$
\item $x_{a,b,c}\geq 0$
\item $x_{a,b,c}=x_{\tau(a),\tau(b),\tau(c)}$ for every permutation
  $\tau$ of $\{a,b,c\}$
\item $x_{a,b,c}\leq t(a,b,c)x_{0,c,c}$
\item $x_{a,b,c}\leq t(b,c,a)x_{0,a,a}$
\item $x_{a,b,c}\leq t(c,a,b)x_{0,b,b}$
\item $\sum_{a,b,c}  T_k(a,b,c)x_{a,b,c}\succeq 0$ for all $0\leq
  k\leq \lfloor \frac{n}{2}\rfloor$
\item $\sum_{a,b,c} T_k(a,b,c)(t(a,b,c)x_{0,c,c}- x_{a,b,c}) \succeq
  0$
for all $0\leq
  k\leq \lfloor \frac{n}{2}\rfloor$
\end{enumerate}
where conditions (vii) and (viii) are equivalent to \eqref{eq1}.
To the above semidefinite constraints we add the extra condition (ix) that
translates the assumption that $d_2(C)\geq d$ for some given value $d$
(respectively $r_2(C)\geq r$, $\rr_2(C)\geq \rr$),
namely
\begin{enumerate}
\item[(ix)$^d$] $x_{a,b,c}=0$\text{ if } $abc\neq 0$ and $a+b+c\leq 2(d-1)$
\end{enumerate}
respectively 
\begin{enumerate}
\item[(ix)$^{\rr}$] $x_{a,b,c}=0$\text{ if } $abc\neq 0$ and $a+b+c < 6\rr$
\end{enumerate}
or
\begin{enumerate}
\item[(ix)$^r$] $x_{a,b,c}=0$\text{ if } $abc\neq 0$ and $\max( \lceil
  \frac{a}{2} \rceil,
\lceil \frac{b}{2} \rceil, \lceil \frac{c}{2} \rceil)\leq r-1$.
\end{enumerate}

\smallskip

\noindent
It remains to notice that
\begin{enumerate}
\item[(x)] $|C|=\sum_{c} x_{0,c,c}$.
\end{enumerate}

Thus an upper bound on $|C|$ is obtained with the optimal value
of the program that maximizes $\sum_{c} x_{0,c,c}$ under the
constraints
(i) to (ix). 

\medskip
It is worth noticing that the conditions (ix) can be replaced
by any other conditions of the type
\begin{enumerate}
\item[(ix$^*$)] $x_{a,b,c}=0$\text{ if } $(a,b,c)\in I$
\end{enumerate}
where $I$ is a set of forbidden values in $C$ related to some other
situation.
In the classical case treated in \cite{S}, $d_1(C)\geq \delta$, 
$I=\{(a,b,c) : a \text{ or }b\text{ or }c\in [1\dots
  (\delta-1)]^3\}$.

\section{Numerical results}\label{tables}

In this section we compare the SDP bounds obtained for $A_2(n,d,m)$
and for $A_2(n,r,m)$
 with the previously known bounds, stated in Section
\ref{bounds}. We recall the obvious values $A_2(n,d,3)= 2^{n-1}$,
$A_2(n,d,n)=4$,
$A_2(n,r,1)=2^n$, $A_2(n,r,\lfloor n/2 \rfloor)=4$.

Table~\ref{tab:d} gives two upper bounds for $A_2(n,d,m)$: one is the
tightest of the combinatorial bounds of section~\ref{bounds}, with a
superscript $1,2,3,4$ denoting which of the four methods, Singleton,
Hamming, Plotkin, Elias (respectively) achieves this best, and the
other is the bound obtained by the SDP method of Section~\ref{SDP}.
As we can see, in the non-trivial cases
the SDP bound gives a substantial improvement almost all the time.

For the radius
$r$, we can restrict ourselves to codes in which the pairwise
distances are even. Let us denote $A_2^+(n,r,m)$ the maximal number of
elements of such a code with minimal radius at least equal to $m$;
then one easily sees that $A_2(n,r,m)=A_2^+(n+1,r,m)$, with the
standard extension of an optimal code to an even code with an extra
coordinate. Table~\ref{tab:r} compares the best bound for $A_2^+(n,r,m)$
(in italics) given by the combinatorial methods of Section~\ref{bounds}
to the SDP bound. Again we have improvements in almost every instance.

%\newpage
 
{\footnotesize
\begin{table}
\begin{tabular}{r|l|l|l|l|l|l|l|l|l|l|l|l|l|l|l|l}
\hspace{-2mm}$n$ \hspace{-1mm}\raisebox{1mm}{$\backslash m$}\hspace{-1mm}
&4&5&6&7&\hspace{-1.5mm} 8&\hspace{-1.5mm} 9&\hspace{-1.5mm} 10&\hspace{-1.5mm}11&\hspace{-1.5mm} 12&\hspace{-1.5mm} 13&\hspace{-1.5mm} 14&\hspace{-1.5mm} 15&\hspace{-1mm}16&\hspace{-1.5mm} 17
\!\!&\hspace{-1.5mm} 18 \!\! &\hspace{-1.5mm}19 \hspace{-2mm}\\
\hline
10 &\hspace{-1mm} 170 &85 &42 &24 &\hspace{-1.5mm} 12 &\hspace{-1.5mm} 6 &&&  &&&& &&&\\
   &\hspace{-1mm}  ${\it 186^2}$& $\it 128^1$& ${\mathit 64^1}$ & 
 $\it 32^1$ &\hspace{-1.5mm} $\it 16^1$&\hspace{-1.5mm} $\it 6^3$
&&&  &&&& &&&
 \\
\hline
11 &\hspace{-1mm}  290&170&85&35&\hspace{-1.5mm} 24&\hspace{-1.5mm} 12&\hspace{-1.5mm} 5&&  &&&& &&& \\
   &\hspace{-1mm} $\it 341^2$ \!\! &$\it  256^1$ \!\! &$\it 128^1$
   \!\! &$\it 61^2$ \!\! &\hspace{-1.5mm} $\it 32^1$ \!\! &\hspace{-1.5mm} $\it
12^3$ &\hspace{-1.5mm} $\it 5^3$
&&&& &&&& &\\
\hline
12 &\hspace{-1mm}  554&277&170&68&\hspace{-1.5mm} 33 &\hspace{-1.5mm} 24&\hspace{-1.5mm} 8&\hspace{-1.5mm} 5&  &&&& &&& \\
   &\hspace{-1mm} $\it 630^2$ \!\! &$\it 512^1$ \!\! &$\it 256^1$ \!\!
   &$\it  103^2$ \!\! &\hspace{-1.5mm} $\it  64^1$ \!\! &\hspace{-1.5mm} $\it 32^1$ 
\!\! &\hspace{-1.5mm} $\it 10^3$ \!\! &\hspace{-1.5mm} $\it 5^3$
&&&& &&&&\\
\hline
13 &\hspace{-1mm} 1042 &521&266&130&\hspace{-1.5mm} 64 &\hspace{-1.5mm} 32&\hspace{-1.5mm} 16&\hspace{-1.5mm} 8&\hspace{-1.5mm} 5 &&&& &&& \\
   &\hspace{-1mm} $\it 1170^2$\! \!\! &$\it 1024^1$\!\! \!\! &$\it
   512^1$\!\! \!\! &$\it  178^2$ \!\! &\hspace{-1.5mm} $\it 128^1$ 
\!\! &\hspace{-1.5mm} $\it 64^1$ \!\! &\hspace{-1.5mm} $\it 32^1$ \!\! &\hspace{-1.5mm} $\it 8^3$ \!\! &\hspace{-1.5mm}$\it 5^3$ &&&& &&&\\
\hline
14 &\hspace{-1mm}  2048&1024&512&257&\hspace{-1.5mm} 128 &\hspace{-1.5mm} 64&\hspace{-1.5mm} 32&\hspace{-1.5mm} 16&\hspace{-1.5mm} 8 &\hspace{-1.5mm}
5&&& &&& \\
   &\hspace{-1mm} $\it 2184^2$\! \!\! &$\it 2048^1$\!\! \!\! &$\it
   1024^1$\!\! \!\! &$\it 309^2$ \!\! &\hspace{-1.5mm} $\it 256^1$ 
\!\! &\hspace{-1.5mm} $\it 128^1$ \!\! &\hspace{-1.5mm} $\it 64^1$ \!\! &\hspace{-1.5mm} $\it 22^3$ \!\!
&\hspace{-1.5mm} $\it 8^3$
\!\! &\hspace{-1.5mm} $\it 5^3$ &&&& &&\\
\hline
15 &\hspace{-1mm} 3616 &2048&1024&414&\hspace{-1.5mm} 256 &\hspace{-1.5mm} 128&\hspace{-1.5mm} 43&\hspace{-1.5mm} 32&\hspace{-1.5mm}
16&\hspace{-1.5mm} 6&\hspace{-1.5mm} 5&& &&& \\
   &\hspace{-1mm} $\it 4096^2$\! \!\! &$\it 4096^1$\!\! \!\! &$\it
   2048^1$\!\! \!\! &$\it 541^2$ \!\! &\hspace{-1.5mm} $\it 512^1$ 
\!\! &\hspace{-1.5mm} $\it 256^1$ \!\! &\hspace{-1.5mm} $\it 113^2$ \!\! &\hspace{-1.5mm} $\it 64^1$ \!\!
&\hspace{-1.5mm} $\it 16^3$
\!\! &\hspace{-1.5mm} $\it 7^3$ 
\!\! &\hspace{-1.5mm} $\it 5^3$&&&&&\\
\hline
16 &\hspace{-1mm} 6963 &3489&2048&766&\hspace{-1.5mm} 382 &\hspace{-1.5mm} 256&\hspace{-1.5mm} 83&\hspace{-1.5mm} 41&\hspace{-1.5mm} 32
&\hspace{-1.5mm} 10&\hspace{-1.5mm} 6&\hspace{-1.5mm} 5& &&& \\
   &\hspace{-1mm} $\it 7710^2$\!\! \!\! &$\it 7710^2$\!\! \!\! &$\it
   4096^1$\!\! \!\! &$\it 956^2$ \!\! &\hspace{-1.5mm} $\it 956^2$ 
\!\! &\hspace{-1.5mm} $\it 512^1$ \!\! &\hspace{-1.5mm} $\it 188^2$ \!\! &\hspace{-1.5mm} $\it 128^1$ \!\!
&\hspace{-1.5mm} $\it 64^1$
\!\! &\hspace{-1.5mm} $\it 13^3$ 
\!\! &\hspace{-1.5mm} $\it 7^3$ \!\! &\hspace{-1.5mm} $\it 5^3$ &&&&\\
\hline
17 &\hspace{-1mm} 13296 &6696&3407&1395&\hspace{-1.5mm} 708&\hspace{-1.5mm} 359&\hspace{-1.5mm} 151&\hspace{-1.5mm} 80&\hspace{-1.5mm} 41
&\hspace{-1.5mm} 20&\hspace{-1.5mm} 10&\hspace{-1.5mm} 6&\hspace{-2mm} 4&&& \\
   &\hspace{-1mm} $\it 14563^2$\!\! \!\! &$\it 14563^2$\!\! \!\! &$\it
   7710^4$\!\! \!\! &$\it 1702^2$\hspace{-1.5mm} &\hspace{-1.5mm} $\it
   1702^2$\hspace{-1.5mm} &\hspace{-1.5mm} $\it 963^4$ \!\! &\hspace{-1.5mm} $\it 314^2$ \!\!
   &\hspace{-1.5mm} $\it
   256^1$ \!\! &\hspace{-1.5mm} $\it 128^1$ \!\! &\hspace{-1.5mm} $\it 52^3$ \!\! &\hspace{-1.5mm} $\it 11^3$ 
\!\! &\hspace{-1.5mm} $\it 6^3$ \!\! &\hspace{-2mm} $\it 4^3$&&&\\
\hline
18 &\hspace{-1mm} 26214  &13107&6555&2559&\hspace{-1.5mm}
1313&\hspace{-1.5mm} 682&\hspace{-1.5mm} 288&\hspace{-1.5mm} 142&\hspace{-1.5mm} 80
&\hspace{-1.5mm} 40&\hspace{-1.5mm} 20&\hspace{-1.5mm} 10&\hspace{-2mm} 6&\hspace{-1mm} 4&& \\
   &\hspace{-1mm} $\it 27594^2$\!\! \!\! &$\it 27594^2$\!\! \!\! &$\it
   15420^4$\!\! \!\! &$\it 3048^2$\hspace{-1.5mm} &\hspace{-1.5mm} $\it
   3048^2$\hspace{-1.5mm} &\hspace{-1.5mm} $\it 1927^4$\hspace{-2mm} &\hspace{-1.5mm} $\it 530^2$ \!\!
   &\hspace{-1.5mm} $\it 512^1$ \!\! &\hspace{-1.5mm} $\it 256^1$ \!\! &\hspace{-1.5mm} $\it 128^1$\hspace{-2mm}
   &\hspace{-1.5mm} $\it 28^3$ 
\!\! &\hspace{-1.5mm} $\it 10^3$ \!\! &\hspace{-2mm} $\it 6^3$ \!\! &\hspace{-1mm} $\it 4^3$\!\!\! &&\\
\hline
19 &\hspace{-1mm} 47337 &26214&13107&4531&\hspace{-1.5mm}
2431&\hspace{-1.5mm} 1284&\hspace{-1.5mm} 513&\hspace{-1.5mm} 276&\hspace{-1.5mm} 142
&\hspace{-1.5mm} 51&\hspace{-1.5mm} 40&\hspace{-1.5mm} 20&\hspace{-2mm} 8&\hspace{-1mm} 6&\hspace{-1mm} 4& \\
   &\hspace{-1mm} $\it 52428^2$\!\! \!\! &$\it 52428^2$\!\! \!\! &$\it
   27594^4$\!\! \!\! &$\it 5489^2$\hspace{-1.5mm} &\hspace{-1.5mm} $\it
   5489^2$\hspace{-1.5mm} &\hspace{-1.5mm} $\it 3246^4$\hspace{-2mm}
   &\hspace{-1.5mm} $\it 903^2$ \!\!
   &\hspace{-1.5mm} $\it 903^2$ \!\! &\hspace{-1.5mm} $\it 512^1$\!\! &\hspace{-1.5mm} $\it 208^2$\hspace{-2mm} &\hspace{-1.5mm}
   $\it 128^1$\hspace{-2mm} &\hspace{-1.5mm} $\it 20^3$ \!\! &\hspace{-2mm} $\it 9^3$ \!\!
&\hspace{-1mm} $\it
6^3$\!\!\! &\hspace{-1mm} $\it 4^3$ \hspace{-2mm} &\\
\hline
20 &\hspace{-1mm} 91750 &46113&26214&8133&\hspace{-1.5mm} 4342&\hspace{-1.5mm} 2373&\hspace{-1.5mm}
1024&\hspace{-1.5mm} 512&\hspace{-1.5mm} 274
&\hspace{-1.5mm} 94&\hspace{-1.5mm} 50&\hspace{-1.5mm} 40&
\hspace{-2mm} 12&\hspace{-1mm} 8&\hspace{-1mm} 6&\hspace{-1.5mm} 4 \\
   &\hspace{-1mm} $\it 99864^2$\!\! \!\! &$\it 99864^2$\!\!\!\! &$\it
55188^4$\!\! \!\! &$\it 9939^2$\hspace{-1.5mm} &\hspace{-1.5mm} $\it
9939^2$\hspace{-1.5mm} &\hspace{-1.5mm} $\it 5518^4$\hspace{-2mm} &\hspace{-1.5mm} $\it 1552^2$\!\!\!
&\hspace{-1.5mm} $\it 1514^4$\hspace{-1.5mm}
&\hspace{-1.5mm} ${\it 1024^1}$\hspace{-1.5mm} &\hspace{-1.5mm} $\it 338^2$\hspace{-2mm} &\hspace{-1.5mm} $\it
256^1$\hspace{-2mm} &\hspace{-1.5mm} $\it 128^1$\!\!\! &\hspace{-2mm} $\it 16^3$\!\!\! 
&\hspace{-1mm} $\it 8^3$\!\!\!&\hspace{-1mm} $\it 6^3$ \hspace{-2mm}
&\hspace{-1.5mm} $\it 4^3$ \hspace{-3mm}\\
\hline
\end{tabular}
\\[0.3cm]
\caption{Bounds on $A_2(n,d, m)$.}
\label{tab:d}
\end{table}

}
%\newpage

{\small
\begin{table}%[htb]
\begin{tabular}{r|c|c|c|c|c|c|c}
&m=2&3&4&5&6&7&8\\
\hline
n=10&96&16&&&&&\\
&\it  102 &\it 22   &&&&&\\
\hline
11&174&26&5&&&&\\
 &\it 186&\it  36&\it  11   &&&&\\
\hline
12&341&48&10&&&&\\
&\it  341&\it  61&\it  17   &&&&\\
\hline
13&582&89&14&5&&&\\
 &\it 630&\it  103&\it  27&\it  10&&&\\   
\hline
14&1109&161&22&5&&&\\
&\it  1170&\it  178&\it  43&\it  14   &&&\\
\hline
15&2085&283&36&9&5&&\\
&\it  2184&\it  309&\it  69&\it  22&\it  9 &&\\  
\hline
16&4096&526&64&13&5&&\\
&\it  4096&\it  541&\it  113&\it  33&\it  13   &&\\
\hline
17&7235&848&123&18&5&4&\\
&\it  7710&\it  956&\it  188&\it  52&\it  19&\it  8&\\   
\hline
18&13926&1550&216&30&10&5&\\
&\it  14563&\it  1702&\it  314&\it  81&\it  27&\it  12&\\   
\hline
19&21883&2852&379&48&12&5&4\\
&\it  27594&\it  3048&\it  530&\it  129&\it  41&\it  16&\it  8   \\
\hline
\end{tabular}\\[0.3cm]
\caption{Bounds on $A_2^+(n,r, m)$.}
\label{tab:r}
\end{table}
}

\end{document}